\documentclass[12pt]{article}

\usepackage{hyperref}
\usepackage{geometry}

\usepackage{amsmath, amscd}
\usepackage{amsfonts}
\usepackage{amssymb}
\usepackage{amsthm}
\usepackage{stmaryrd}

\DeclareMathOperator{\ad}{ad}
\newcommand{\rd}{\overleftarrow{\partial}} 
\newcommand{\ld}{\overrightarrow{\partial}} 
\newtheorem{definition}{Definition}[section]
\newtheorem{theorem}{Theorem}[section]
\newtheorem{lemma}{Lemma}[section]

\newtheorem{example}{Example}[section]

\begin{document}

\title{{Lectures on Deformation quantization of Poisson manifolds}}
\author{Chiara Esposito\\
\texttt{esposito@math.ku.dk}\\
\texttt{esposito.chiar@gmail.com}
}

\maketitle

\abstract{In these notes we review the Kontsevich's formality therem as discussed at the \textit{School of Geometry, University Kasdi Merbah (Ouargla) 2012}. They are essentially based on \cite{CI}, where the interested reader can find many details we are not going to discuss.}

\clearpage

\tableofcontents

\vfill
\section*{Acknowledgments}
We thank the hospitality 
at University Kasdi Merbah (Ouargla).
Thanks to Mohamed Amine Bahayou for the nice organization and for the support during our stay in Algeria.
We also want to thank all the participants to the school for 
the enthusiasm and all the nice questions that helped us to write these notes. A special thank to Peter Bongaart for his interests in these lectures and all the stimulating discussions we have had.

\clearpage 

\section{Introduction} 

The subject of these lecture notes is the theory of Kontsevich of deformation quantization of Poisson manifolds. We start a brief survey of tensor fields that will be useful to define Poisson manifolds. In Section \ref{sec:pm} we give an overview of physical motivations that led to the introduction of the star product and we summarize the steps in the theory of deformation quantization. Sections \ref{sec:dqp} and \ref{sec:ft} are devoted to Kontsevich's theory. First, we introduce the theory of classification of star products that led to claim that any Poisson manifold admits a canonical deformation quantization. In order to prove this result we need to introduce a more general one, called Formality theorem. 

\section{Basic notions}

In this section we recall briefly the definition of tensor field that will be useful to introduce the Poisson manifolds.
A more complete discussion on tensors and all the basic theory about manifolds and vector bundles can be found in \cite{MAR}.  

Consider the set $L^k(V_1,\dots, V_k; W)$ of $k$-multilinear maps of $V_1\times\dots V_k$ to $W$. The special case $L(V, \mathbb{R})$ is denoted $V^*$, the dual space of $V$. If $V$ is finite dimensional and $\lbrace e_1, \dots e_n\rbrace$ is a basis of $V$, there is a unique basis of $V^*$, the dual basis $\lbrace f^1, \dots f^n\rbrace$, such that $\langle f^i,e_j\rangle= \delta^i_j$. Here $\langle \cdot,\cdot\rangle$ denotes the pairing between $V$ and $V^*$.
\begin{definition}
For a vector space $V$ we put
$$
T_s^r(V)=L^{s+r}(V^*,\dots, V^*, V,\dots, V; \mathbb{R})
$$
($r$ copies of $V^*$ and $s$ copies of $V$). Elements of $T_s^r(V)$ are called \textbf{tensors} on $V$, contravariant of order $r$ and covariant of order $s$.

Given $t\in T_s^r(V)$ and $s\in T_p^q(V)$, the \textbf{tensor product} of $t$ and $s$ is the tensor $t\otimes s\in T_{s+p}^{r+q}(V)$ defined by
\begin{align}
(t\otimes s)&(\beta^1, \dots \beta^r, \gamma^1, \dots, \gamma^q, f_1, \dots f_s,g_1, \dots, g_p)\\
&=t(\beta^1, \dots \beta^r,f_1, \dots f_s)s(\gamma^1, \dots, \gamma^q, g_1, \dots, g_p)
\end{align}
where $\beta^j, \gamma^j\in V^*$ and $f_j, g_j\in V$.
\end{definition}
The tensor product is associative, bilinear and continuous; it is not commutative. Notice that
$$
T^1_0(V)=V, \quad T^0_1(V)=V^*.
$$
Now we extend the tensor algebra to vector bundles.
\begin{definition}
Let $\pi: V\rightarrow B$ be a vector bundle with $V_b=\pi^{-1}(b)$ denoting the fiber over the point $b\in B$. Define
$$
T^r_s(V)=\bigcup_{b\in B}T^r_s(V_b)
$$
and $\pi_s ^r: T_s^r(V)\rightarrow B$ by $\pi_s^r(v)=b$, where $e\in T^r_s(V_b)$. Furthermore, for a given subset $A$ of $B$, we define 
$$
T^r_s(V)\vert_A=\bigcup_{b\in A}T^r_s(V_b).
$$
\end{definition}
Let us consider a smooth manifold $M$. We denote with $C^{\infty}(M)$ the set of all the smooth functions from $M$ to $\mathbb{R}$. We specialize to the case where $\pi: V\rightarrow B$ is the tangent vector bundle of $M$.

\begin{definition}
Let $M$ be a manifold and $\tau_M:TM\rightarrow M$ its tangent bundle. We call $T_s^r(M)=T_s^r(TM)$ the \textbf{vector bundle of tensors} contravariant order $r$ and covariant order $s$. We identify $T^1_0(M)$ with $TM$ and call $T_1^0(M)$ the cotangent bundle of $M$ also denoted by $\tau_M^*: T^*M\rightarrow M$. The zero section of $T_s^r(M)$ is identified with $M$.
\end{definition}

The smooth sections of $\pi: V\rightarrow B$ are denoted by $\Gamma(V)$.
A section of $T_s^r(M)$ takes an element $m\in M$ and associates a vector in the fiber, called tensor. Recall that the set of smooth functions $C^{\infty}(M)$ is endowed with a structure of ring, defined by
\begin{equation}
(f+g)(x)=f(x)+g(x), \quad (cf)(x)=c(f(x)), \quad (fg)(x)=f(x)g(x).
\end{equation}
Finally, recall that a vector field on $M$ is an element of $\Gamma(TM)$.

\begin{definition}
A \textbf{tensor field} of type $(r,s)$ on a manifold $M$ is a smooth section of $T_s^r(M)$. We denote by $\mathcal{T}_s^r(M)$ the set $\Gamma(T_s^r(M))$ together with its infinite dimensional real vector space structure. A covector field or differential one-form is an element of $\mathcal{T}_1^0(M)$. 
\end{definition}

\subsection{Poisson manifolds}

Now we introduce the definition of Poisson manifold, in terms of algebra and then in a more general way using the structures discussed in the previous section. This section is essentially based on the book of  I. Vaisman \cite{V}. We want to remark that Poisson manifolds have many geometrical properties and there is an incredible rich literature devoted on this topic. Here we just recall their definition and a simple example in order to discuss their deformation quantization.

\begin{definition}
 A \textbf{Poisson manifold} is a pair $(M,\left\lbrace \cdot,\cdot\right\rbrace )$, where $M$ is a smooth manifold and $\left\lbrace \cdot,\cdot\right\rbrace$ is a bilinear operation on $C^{\infty}(M)$, such that the pair $(C^{\infty}(M),\left\lbrace \cdot,\cdot\right\rbrace)$ is a Lie algebra and $\left\lbrace \cdot,\cdot\right\rbrace$ is a derivation in each argument. The pair $(C^{\infty}(M),\left\lbrace \cdot,\cdot\right\rbrace)$ is called Poisson algebra. 
\end{definition}
Let $\pi$ be a bivector field on a manifold $M$, i.e. a skew-symmetric, contravariant 2-tensor. At each point $m$, $\pi_m$ can be viewed as a skew-symmetric bilinear form on $T_m^*M$, or as the skew-symmetric linear map 
$\pi^{\sharp}_m: T_m^*M\rightarrow T_mM$, such that
\begin{equation}\label{eq: pish}
  \pi_m(\alpha_m,\beta_m)=\pi^{\sharp}(\alpha_m)(\beta_m), \quad \alpha_m,\beta_m\in T_m^*M.
\end{equation}
If $\alpha$, $\beta$ are 1-forms on $M$, we define $\pi(\alpha,\beta)$ to be the function in $C^{\infty}(M)$ whose value at $m$ is $\pi_m(\alpha_m,\beta_m)$. Given $f,g\in C^{\infty}(M)$ we set
\begin{equation}\label{eq: pbra}
\lbrace f,g\rbrace_m =  \pi_m(df,dg).
\end{equation}
The bracket induced by $\pi$ satisfies the Leibnitz rule.

\begin{definition}
 A \textbf{Poisson manifold} $(M,\pi)$ is a manifold $M$ with a Poisson bivector $\pi$ such that the bracket defined in eq. (\ref{eq: pbra}) satisfies the Jacobi identity.
\end{definition}

\begin{example}\label{r2}
 If $M=\mathbb{R}^{2n}$, with coordinates $(q^i, p_i)$, $i=1,\cdots, n$ and if
\begin{equation}
 \pi^{\sharp}(dq^i)=-\frac{\partial}{\partial q^i}, \quad  \pi^{\sharp}(dp_i)=-\frac{\partial}{\partial p_i},
\end{equation}
then
\begin{equation}
 X_f=\frac{\partial f}{\partial p_i}\frac{\partial}{\partial q^i}-\frac{\partial f}{\partial q^i}\frac{\partial}{\partial p_i}
\end{equation}
and
\begin{equation}
\lbrace f,g\rbrace=\frac{\partial f}{\partial p_i}\frac{\partial g}{\partial q^i}-\frac{\partial f}{\partial q^i}\frac{\partial g}{\partial p_i},
\end{equation}
is the standard Poisson bracket of functions on the phase space. The corresponding bivector is $\pi=\frac{\partial }{\partial p_i}\wedge \frac{\partial}{\partial q^i}$.
\end{example}

\section{Physical motivation}\label{sec:pm}
In this section we want to describe briefly the physical motivations that gave rise to the theory of deformation quantization. Essentially this theory want to give a precise mathematical formulation to the correspondence between classical and quantum mechanics. The first step is recalling the mathematical description of such theories and then we discuss the attempts to describe this correspondence. It is important to remark that a formal correspondence between the two theories is still missing, despite the fact that many progresses in that direction have been done.

A classical mechanical system in the hamiltonian formalism is described by the triple $(M,\lbrace\cdot,\cdot\rbrace, H)$ where $M$ is an even-dimensional manifold called phase space, $\lbrace\cdot,\cdot\rbrace$ is a Poisson bracket, induced equivalently by a symplectic or Poisson structure on $M$ and $H$ is a smooth function on $M$, called Hamiltonian. These three objects allow us to describe completely a given physical system. Indeed, a physical state of the system is represented by a point of $M$ and a physical observable corresponds to a smooth function $f$ on $M$. The Poisson bracket and the Hamiltonian function are necessary to describe the time evolution of an observable $f$, that is governed by the equation
\begin{equation}
\frac{df}{dt}=\lbrace H,f\rbrace
\end{equation}
Here the Poisson brackets is completely determined by its action on
the coordinate functions
\begin{equation}
\lbrace q_i,p_j\rbrace= \delta_{ij}
\end{equation}
(together with $ \lbrace q_i,q_j\rbrace= \lbrace p_i,p_j\rbrace= 0$)
where $(q_1, \ldots , q_n, p_1, \ldots , p_n)$ are 
local coordinates on the $2n$- dimensional manifold $M$.

In quantum mechanics, a physical system can be described
by a complex Hilbert space $\mathcal{H}$ together with an Hamiltonian operator $\widehat{H}$.
In this formalism a physical state is represented by a vector in $\mathcal{H}$ while the physical observables are now 
self-adjoint operators in the space $\mathcal{L}(\mathcal{H})$ of linear operators on $\mathcal{H}$. 
The time evolution of a physical observable $\hat{f}$ is given by
\begin{equation}
\frac{d\hat{f}}{dt} = \frac{i}{\hbar}\left[ \hat{H},\hat{f}\right] 
\end{equation}
where $[\cdot,\cdot]$ is the usual commutator of operators. Finally,  the position $\hat{q}_i$ 
and momentum $\hat{p}_j$ operators satisfy
the canonical commutation relations:
\begin{equation}
[\hat{q}_i,\hat{p}_j] = i\hbar \delta_{ij}.
\end{equation}

The correspondence between classical and quantum mechanics has been studied from many different points of view. Now we focus our attention on the attempt of finding a precise mathematical procedure to associate to a classical observable a quantum analog. 
This was first approached by trying to construct a correspondence between the commutative algebra $C^{\infty}(M)$ and the non-commutative algebra of operators. Starting from the quantization of $\mathbb{R}^{2n}$, the first result was achieved by Groenewold \cite{Gro}, which states that  the Poisson algebra $C^{\infty}(\mathbb{R}^{2n})$ can not be quantized in such a way that the Poisson bracket of two classical observables is mapped into the Lie bracket of the correspondent operators. 

The idea of Bayen, Flato et al. \cite{BFLS},\cite{FLS1}, \cite{FLS2} was a change of perspective: instead of mapping functions to operators, the algebra of functions can be deformed into a non-commutative one.
In particular, they proved
that on the symplectic vector space $\mathbb{R}^{2n}$, there exists a standard deformation quantization, or star product, known as the Moyal-Weyl product. 
The origins of the Moyal-Weyl product can be found in the works of Weyl \cite{Wy} 
and Wigner \cite{Wi}, where they give an explicit correspondence between functions and operators, and of Groenewold \cite{Gro} and Moyal \cite{Mo}, where 
the product and the bracket of operators defined by Weyl have been introduced. The existence of an associative star product has been generalized to a symplectic manifold admitting a flat connection $\nabla$ in \cite{BFLS}. 
The first proof of the existence of star product for any symplectic manifold was given by De Wilde and Lecomte \cite{DL} and few years later by Fedosov \cite{Fe}. In subsequent works (e.g. \cite{NT}, \cite{Gu}) the equivalence classes of star products on symplectic manifolds and the connection with de Rham cohomology has been studied. It came out that the equivalence classes of star products and elements in $H_{dR}^2(M)  \llbracket \epsilon  \rrbracket $ are in a one-to-one correspondence.

The existence and classification of star product culminated with Kontsevich's Formality Theorem, that was first formalized in a conjecture in \cite{Ko1} and then proved in \cite{Ko}. Kontsevich showed that any finite dimensional Poisson manifold $M$ admits a canonical deformation quantization by establishing a correspondence between the set of isomorphism classes of deformations of $C^{\infty}(M)$ and the set of equivalence classes of formal Poisson structures on $M$. In the following sections we introduce the classification of star products proved by Kontsevich and then we show how this result follows from the Formality theorem.

\section{Deformation quantization of Poisson manifolds}\label{sec:dqp}

In this section we introduce the basic notion of formal deformation of an algebra $A$ and then we explain the connection of deformations with Poisson structures. The first step to discuss the theory of classification of star products is the introduction of a new tool, the formal power series.

\subsection{Formal power series}

Given a sequence $a_n$, $n\in \mathbb{N}_0$ of elements in a commutative ring $k$, a formal power series is
\begin{equation}
a=\sum_{n=0}^{\infty}a_n x^n
\end{equation}
where $x$ is just a formal indeterminate, that means that we look at power series as purely algebraic objects, which we manipulate according to some set of rules. In particular, we are not interested in the analytic properties of the formal power series, for instance we will not require the convergence. 

They can be manipulate algebrically:
\begin{align}
a\pm b &=\sum_{n=0}^{\infty}(a_n\pm b_n)x^n\\
ab &=\sum_{n=0}^{\infty}c_n x^n\qquad c_n=\sum_{k=0}^n a_kb_{n-k}
\end{align}
With these two operations, the set of all formal power series becomes a commutative ring, denoted by $k  \llbracket x \rrbracket $.
If we can prove the convergence in a neighborhood of zero, we can see that the coefficients are the same as the Taylor expansion (see \cite{FS} for more details).

\subsection{Classification of star product}\label{sec:cl}

The idea of star product relies on the more basic definition of formal deformation of an algebra given by Gerstenhaber \cite{G}.
\begin{definition}
Let $A$ be an associative and unital algebra over a commutative ring $k$. A \textbf{formal deformation} of the algebra $A$ is a formal power series 
\begin{equation}\label{eq: star}
a\star b=ab+\sum_{k=1}^{\infty}\hbar^k P_k(a,b)
\end{equation}
where $a,b\in A\subset A  \llbracket \hbar  \rrbracket$ and $P_m:A\times A\rightarrow A$ are $k$-bilinear maps such that
the product $\star$ is associative.
\end{definition}
The deformed algebra over the ring $k  \llbracket \hbar  \rrbracket $ is denoted by $A  \llbracket \hbar  \rrbracket $.

Let $M$ be a smooth manifold and consider the algebra $C^{\infty}(M)$ of smooth functions on $M$ endowed with the pointwise product
\begin{equation}
f\cdot g(x):=f(x)g(x)\qquad \forall x\in M.
\end{equation}
In this case the ring $k$ is $\mathbb{R}$. A deformation quantization of $C^{\infty}(M)$ is a formal deformation of $A=C^{\infty}(M)$ such that it preserves the unit of the algebra. Let's state it more clearly:

\begin{definition}\label{def-star}
A \textbf{star product} on $M$ is an $\mathbb{R}\llbracket \hbar  \rrbracket$-bilinear map
\begin{align}\nonumber
C^{\infty}(M)\llbracket \hbar  \rrbracket\times C^{\infty}(M)\llbracket \hbar  \rrbracket &\rightarrow C^{\infty}(M)\llbracket \hbar  \rrbracket\\
(f,g) &\mapsto f\star g
\end{align}
such that
\begin{enumerate}
    \item $f\star g=f\cdot g+\sum_{k=1}^{\infty}\hbar^k P_k(f,g)$
	\item $(f\star g)\star h=f\star (g\star h)$    $\forall f,g,h\in C^{\infty}(M)$   associativity
	\label{cond-un}\item $f\star 1=1\star f= f$    $\forall f\in C^{\infty}(M)$.
\end{enumerate}
\end{definition}
The requirement (\ref{cond-un}) implies that the $P_n$'s are bidifferential operators.

\begin{example}
The first example of a deformed product on $C^{\infty}(\mathbb{R}^{2n})$ is the \textbf{Moyal product}. We introduce briefly this example and we show the relation of the Moyal product with the standard Poisson bracket.
Let us consider the manifold $M=\mathbb{R}^{2n}$ with Darboux coordinates 
$$(q,p)=(q_1,\dots,q_n,p_1,\dots, p_n)$$
 The Moyal star product on this manifold is given by
\begin{equation} \label{moyal}
f\star g(q,p)= f(q,p)\; \exp \left(i \frac{\hbar}{2}
\left( \rd_q \ld_p -  \rd_p \ld_q \right) \right)\; g(q,p), 
\end{equation} 
where the $\rd$'s operate on $f$ and the $\ld$'s on $g$; 
we can also define a star product by
\begin{equation} \label{gen-moyal}
f \star g\; (x) = \exp \left(i \frac{\hbar}{2}\; \alpha^{ij} \partial_{x^i} \partial_{y^j}  \right) 
f(x)\, g(y)\Big|_{y=x}.
\end{equation}
where $\lbrace\alpha^{ij}\rbrace$ is a constant skew-symmetric tensor
 on $\mathbb{R}^2n$ with $i,j=1, \ldots, n$.
Such a product satisfies the requirements given in Definition \ref{def-star}, hence it is a well defined star product. Observe that the skew-symmetric part of the term in $\hbar$ of \ref{moyal} is
\begin{equation}
\frac{\partial f}{\partial p_i}\frac{\partial g}{\partial q^i}-\frac{\partial f}{\partial q^i}\frac{\partial g}{\partial p_i},
\end{equation}
which coincides with the standard Poisson bracket introduced in Example \ref{r2}. In this particular case we can claim that defining the Moyal product on $C^{\infty}(\mathbb{R}^{2n})$, it inherits automatically a structure of Poisson algebra.
\end{example}

The last observation can be generalized. From the requirements on $P_n$ it follows that the skew-symmetric part of $P_1$ defined by
\begin{equation}\label{eq:ps}
\lbrace f,g\rbrace=P_1(f,g)-P_1(g,f)\qquad f,g\in A
\end{equation}
is a Poisson bracket, then $A$ is a Poisson algebra. If $A=C^{\infty}(M)$, this implies that $M$ is a Poisson manifold, with Poisson structure $\pi$ such that
$$
\lbrace f,g\rbrace=\pi(df,dg).
$$

Now we want to understand if, given a Poisson manifold $(M,\pi)$, 
we can define an associative product $\star$ on the algebra of smooth functions using the structure $\pi$, i.e. such that:
\begin{enumerate}
	\item it is a deformation of the pointwise product
	\item $\frac{1}{\hbar} (f \star g - g \star f)\mod \hbar= \lbrace f,g\rbrace$.
\end{enumerate}

This problem has been solved by Kontsevich, by classifying star products in terms of Poisson structures. In order to introduce this result we need the concept of equivalent star products. 

\begin{definition}\label{def:iso}
An \textbf{isomorphism} of two deformations $\star$, $\star'$ is a formal power series \linebreak $T(a)=a+\sum_{m=0}^{\infty}t^m T_m(a)$ such that
\begin{equation}\label{iso}
 T(a\star b)=T(a)\star' T(b)\quad\forall a,b\in A.
\end{equation}
The star products $\star$ and $\star'$ are said equivalent.
\end{definition}

From now on we denote with $[\star]$ the equivalence class of star products relative to the previous definition of equivalence.
We can prove that different star products belonging to the same equivalence class induce the same Poisson bracket by setting (\ref{eq:ps}). More precisely,

\begin{lemma}
Let $\star$ be a star product on $C^{\infty}(M)$. The Poisson bracket
$$
\lbrace f,g\rbrace=P_1(f,g)-P_1(g,f)\qquad f,g\in C^{\infty}(M)
$$
 depends only on the equivalence class $[\star]$.
\end{lemma}
\begin{proof}
Consider two equivalent star products $\star$ and $\star'$. From Definition \ref{def:iso} we have
$$
T(a\star b)=T(a)\star' T(b)
$$
Expanding the formal power series of $T$, $\star$ and $\star'$, the term in $\hbar$ of this equation reads
$$
P_1(f,g)+T_1(fg)=P'_1(f,g)+T_1(f)g+fT_1(g).
$$
This implies that $P_1(f,g)-P'_1(f,g)$ is symmetric in $f, g$, hence it does not contribute to $\lbrace f,g\rbrace$.
\end{proof}

This implies that, given an equivalence class of star products on $C^{\infty}(M)$, it induces a Poisson structure $\pi$ on the manifold $M$. 
As introduced in Section \ref{sec:pm}, the problem of classifying star products on a given Poisson manifold $M$ is solved by proving that there is a one-to-one correspondence between equivalence classes of star products and equivalence classes of formal Poisson structures. In the following we define the equivalence relation of formal poisson structures.

We can define a group of diffeomorphisms of $M$ acting on the set of Poisson structures, given by
\begin{equation} \label{phi-on-pi}
\pi_{\phi} := \phi_{*} \pi.
\end{equation}
This action can be easily extended to formal power series; let us introduce
 a bracket
on $C^{\infty}(M)\llbracket \hbar  \rrbracket$ by:
\begin{equation} \label{formal-poisson}
\lbrace f,g\rbrace_{\hbar}:= \sum_{n=0}^{\infty} \hbar^n   \sum_{\substack{
  i,j,k=0 \\
    i+j+k=n
  }}^n\pi_i(d f_j , d g_k)
\end{equation}
where 
$$
f = \sum_{j=0}^{\infty} \hbar ^j f_j
\qquad\text{and}\qquad 
g = \sum_{k=0}^{\infty} \hbar^k g_k
$$
The structure
$$
\pi_{\hbar}:= \pi_0 + \pi_1 \, \hbar + \pi_2 \, \hbar^2 + \cdots
$$
is called \textbf{formal Poisson structure} if $\lbrace\cdot,\cdot\rbrace_{\hbar}$ is a Lie bracket 
on $C^{\infty}(M)\llbracket \hbar  \rrbracket$.

The gauge group is given by 
the formal power series of the form
$$
\phi_{\hbar}:= \exp(\hbar X)
$$
called\textbf{formal diffeomorphisms},
where $X:=\sum_{k=0}^{\infty} \hbar^k X_k $ is a \textbf{formal vector field}, i.e.\
a formal power series whose coefficients are vector fields. It is useful to remark that the structure of a group
is given by defining the product 
of two such exponentials via the Baker-Campbell-Hausdorff formula:
\begin{equation} \label{bch}
\exp(\hbar X) \cdot  \exp(\hbar Y) :=
\exp(\hbar X + \hbar Y + \frac{1}{2} \hbar [X,Y] + \cdots).
\end{equation}
Hence we can generalize the action \eqref{phi-on-pi} as follows:
\begin{equation} \label{gauge-action}
\exp (\mathcal{L}_X)\pi_{\hbar}=\exp ( \hbar X)_{*} \pi_{\hbar} := 
\sum_{n=0}^{\infty} \hbar^n 
\sum_{\substack{
  i,j,k=0 \\
    i+j+k=n
  }}^n (\mathcal{L}_{X_i})^j \pi_k
\end{equation}

\begin{definition}
Two formal Poisson structures $\pi_{\hbar}$ and $\pi'_{\hbar}$ are said \textbf{equivalent} if there is a formal power series $X=\sum_{m=0}^{\infty}\hbar^m X_m$ such that 
$$\pi'_{\hbar}=\exp(\mathcal{L}_X)\pi_{\hbar}.$$ 
\end{definition}
The equivalence class is denoted by $[\pi_{\hbar}]$, as usual. We can finally state the Kontsevich theorem:

\begin{theorem}[Kontsevich, \cite{Ko}]\label{thm:k}
There is a bijection
$$
[\pi_{\hbar}]\longleftrightarrow [\star]
$$
 natural with respect to diffeomorphisms, between the set of equivalence classes $[\pi_{\hbar}]$ of formal Poisson structures  on $M$
 \begin{equation} \label{eq:phb}
\pi_{\hbar}=0+\hbar\pi_1+\hbar^2\pi_2+\dots
\end{equation}
 and the set of isomorphism classes $[\star]$ of deformation quantizations of $C^{\infty}(M)$.
\end{theorem}

 Moreover, if $\pi_{\hbar}$ in (\ref{eq:phb}) is a formal Poisson structure, we will denote by $\star_{\hbar}$ a star product from the equivalence class corresponding to $\pi_{\hbar}$ by the above theorem. The Poisson bracket $\lbrace\cdot,\cdot\rbrace$ on $C^{\infty}(M)$ associated to $\star_{\hbar}$ by formula (\ref{eq:ps}) is given by the term $\pi_1$ of $\pi_{\hbar}$.
This implies that any Poisson manifold $(M,\pi)$ admits a canonical deformation quantization, that is the quantization obtained applying this theorem to $\pi_{\hbar}=\hbar\pi$.
This result follows from a more general one, called Formality theorem, that we introduce in the following section.

\section{Formality Theory}\label{sec:ft}

In this section we show that to each deformation is attached a differential graded Lie algebra via the solutions to the Maurer Cartan equation modulo the action of a gauge group. In order to discuss this problem, we need to introduce the definition of differential graded Lie algebras and some properties. We focus our attention of the differential graded Lie algebras of multidifferential operators and multivector fields and, by means of the concept of $L_{\infty}$ morphism, we introduce the Formality theorem. Finally we give a sketch of the proof of Theorem \ref{thm:k}. More details can be found in \cite{Arb}, \cite{Ma}, \cite{CI}.

\subsection{Differential Graded Lie Algebras}
\begin{definition}\label{gla}
A \textbf{graded Lie algebra} (GLA) is a graded vector space \linebreak 
$\mathfrak{g}=\oplus_{i\in\mathbb{Z}}\mathfrak{g}^i$ endowed with a bilinear operation
\begin{equation}
[\cdot,\cdot]:\mathfrak{g}\otimes\mathfrak{g}\rightarrow \mathfrak{g}
\end{equation}
satisfying the following conditions:
\begin{enumerate}
\item homogeneity, $[a,b]\in \mathfrak{g}^{\alpha+\beta}$
\item skew-symmetry, $[a,b]=-(-1)^{\alpha\beta}[b,a]$
\item Jacobi identity, $[a,[b,c]]=[[a,b],c]+(-1)^{\alpha\beta}[b,[a,c]]$
\end{enumerate}
for any $a\in \mathfrak{g}^{\alpha}$, $b\in \mathfrak{g}^{\beta}$ and $c\in \mathfrak{g}^{\gamma}$.
\end{definition}

As an example, any Lie algebra is a GLA concentrated in degree 0.

\begin{definition}
A \textbf{differential graded Lie algebra} (DGLA) is a GLA $\mathfrak{g}$ together with a differential $d:\mathfrak{g}\rightarrow \mathfrak{g}$, i.e. a linear operator of degree 1 which satisfies the Leibnitz rule
\begin{equation}
d[a,b]=[da,b]+(-1)^{\alpha\beta}[a,db]\qquad a\in \mathfrak{g}^{\alpha},\quad b\in \mathfrak{g}^{\beta}
\end{equation}
 and $d^2=0$.
\end{definition}
Given a DGLA we can define immediately the cohomology \footnote{here we don't give basic definitions on cohomology theory; an excellent introduction can be found in \cite{BT}} of $\mathfrak{g}$ as
\begin{equation}
H^i(\mathfrak{g}):=Ker(d:\mathfrak{g}^i\rightarrow \mathfrak{g}^{i+1})/Im(d: \mathfrak{g}^{i-1}\rightarrow \mathfrak{g}^i)
\end{equation}
The set $H:=\oplus_i H^i(\mathfrak{g})$ has a natural structure of graded Lie algebra (because of the compatibility between $d$ and $[\cdot,\cdot]$ on $\mathfrak{g}$, it inherits the GLA structure defined on equivalence classes $\vert a\vert, \vert b\vert\in H$ by
$$
\left[ \vert a\vert, \vert b\vert\right] _H:=\vert [a,b]\vert.
$$
Finally, $H$ is a DGLA by putting $d=0$.)

A morphism of DGLA is a linear homogeneous map $f:\mathfrak{g}_1\rightarrow \mathfrak{g}_2$ of degree zero, such that
$$
f\circ d=d\circ f
$$
and
$$
f([x,y])=[f(x),f(y)].
$$
The morphism $f:\mathfrak{g}_1\rightarrow \mathfrak{g}_2$ of DGLA's induces a morphism $H(f):H_1\rightarrow H_2$ between cohomologies (i.e. the sequence of homomorphisms $H^n(f):H^n(\mathfrak{g}_1)\rightarrow H^n(\mathfrak{g}_2)$ ). A \textbf{quasi-isomorphism} is a morphism of DGLA's inducing isomorphisms in cohomology.

\begin{definition}\label{def: fdgla}
A differential graded Lie algebra $\mathfrak{g}$ is \textbf{formal} if it is quasi-isomorphic to its cohomology, regarded as a DGLA with zero differential and the induced bracket.
\end{definition}
The result of Kontsevich that we are going to introduce is called \textbf{formality theorem} because it shows the the DGLA of multidifferential operators, that we are going to define, is formal.

We already stated that to each deformation is attached a differential graded Lie algebra via the solutions to the Maurer Cartan equation modulo the action of a gauge group. 
We introduce now the Maurer-Cartan equation of a DGLA and the gauge group defined for any formal DGLA by generalizing what we did in Section \ref{sec:cl} with $\pi_{\hbar}$. 

The Maurer-Cartan equation of the DGLA $\mathfrak{g}$ is
\begin{equation} \label{MC-eq}
d a + \frac{1}{2} [a,a] = 0 \qquad a \in \mathfrak{g}^1,
\end{equation}
We can define a gauge group acting on the solutions of the Maurer-Cartan equation starting 
from the degree zero part of any formal DGLA. Indeed, given a DGLA $\mathfrak{g}$ we can define its formal counterpart
$ \mathfrak{g}\llbracket \hbar  \rrbracket$ by
$$
\mathfrak{g}\llbracket \hbar  \rrbracket:=\mathfrak{g}\otimes k\llbracket \hbar  \rrbracket;
$$
it has the natural structure of a DGLA.
It is clear that the degree zero part $\mathfrak{g}^0\llbracket \hbar  \rrbracket$ is a Lie algebra. 

As seen in the case of formal Poisson structures, we can define the gauge group formally as the set 
$$G:= \exp( \hbar \mathfrak{g}^0\llbracket \hbar  \rrbracket)$$ 
and introduce a well-defined 
product taking the Baker-Campbell-Hausdorff formula \ref{bch}.  
Finally, the action of the group on $\hbar \mathfrak{g}^1\llbracket \hbar  \rrbracket$ can 
be defined generalizing the adjoint action in \ref{gauge-action},
namely:
\begin{align}\nonumber
\exp(\hbar  g) a &:= \sum_{n=0}^{\infty} \frac{(\ad g)^n}{n !} (a) -
\sum_{n=0}^{\infty} \frac{(\ad g)^n}{(n+1) !}(d g) \\
\label{gac} &=a + \hbar [g,a] - \hbar d g + o(\hbar^2)
\end{align}

for any $g \in  \mathfrak{g}^0\llbracket \hbar  \rrbracket$ and $a \in  \mathfrak{g}^1\llbracket \hbar  \rrbracket$.
It is easy to show that this action preserves
the subset $MC(\mathfrak{g}) \subset \hbar\mathfrak{g}^1\llbracket \hbar  \rrbracket$ of solutions to the 
(formal) Maurer-Cartan equation. We will discuss explicitly the Maurer-Cartan equation's solutions and 
the gauge group action in the case of multidifferential operators and multivector fields.

\subsection{Multivector fields and multidifferential operators}
The Kontsevich's theorem, as we know, proved a correspondence between Poisson structures and star product; in order to prove this correspondence we introduce the DGLA's they belong to.
\subparagraph{Multivector fields}

 By definition, a $k$-multivector field $X$ is a section of the $k$-th exterior power $\wedge^k TM$ of the tangent space $TM$. In local coordinates $\lbrace x_i\rbrace_{i=1}^m$, the multivector field $X\in \Gamma(M,\wedge^k TM)$ can be written as
 \begin{equation}
X=\sum_{i_1\dots i_k=1}^m X^{i_1\dots i_k}(x)\partial_{i_1}\wedge\dots\wedge\partial_{i_k}.
\end{equation}
It is evident that $\Gamma:=\oplus_{k=0}^{\infty}\Gamma^k$ is a graded vector space, where
\begin{equation}
\Gamma^k =\begin{cases} C^{\infty}(M), & \mbox{if }k=0 \\\Gamma(M,\wedge^k TM), & \mbox{if } k\geq1
\end{cases}
\end{equation}
The Lie algebra structure is given by the \textbf{Schouten-Nijenhuis bracket} $[\cdot,\cdot]_S:\Gamma^k\otimes \Gamma^l\rightarrow \Gamma^{k+l-1}$ defined by
\begin{equation}\nonumber
\begin{split}
&[X_1\wedge\dots\wedge X_k,Y_1\wedge\dots\wedge Y_l]_S:=\\
&\sum_{i=1}^k\sum_{j=1}^l(-1)^{i+j}[X_i,Y_j]\wedge X_1\wedge\dots\wedge\hat{X}_i\wedge\dots\wedge X_k\wedge Y_1\wedge\dots\wedge\hat{Y}_j\wedge\dots\wedge Y_l.
\end{split}
\end{equation}
This bracket satisfies the following properties
\begin{itemize}
\item[i)] $[X,Y]_S = - (-)^{(x+1)(y+1)} [Y,X]$
\item[ii)] $[X,Y \wedge Z] = [X,Y]\wedge Z+ (-)^{(y+1) z} Y\wedge[X,Z]$
\item[iii)] $[X,[Y,Z]] = [[X,Y],Z]+ (-)^{(x+1)(y+1)} [Y,[X,Z]]$
\end{itemize}
for any triple $X$,$Y$ and $Z$ of degree resp. $x$, $y$ and $z$. In order to recover the sign used in Definition \ref{gla}, we shift the degree
\begin{equation}
\tilde{\Gamma}:=\oplus_{i=-1}^{\infty}\tilde{\Gamma^i}\qquad\text{where } \tilde{\Gamma^i}:=\Gamma^{i+1}
\end{equation}
The GLA $\tilde{\Gamma}$ is turned into a DGLA setting the differential
$d:\tilde{\Gamma}\rightarrow \tilde{\Gamma}$ to be identically zero.
We denote this DGLA by $\mathfrak{g}_S^{\bullet}(M)$.

We now focus to the particular class of Poisson bivector fields.
Recall that given a bivector field $\pi \in \mathfrak{g}_S^{1}(M)$, we can  define a Poisson bracket by
\begin{equation}
\lbrace f,g\rbrace=\pi(df,dg)
\end{equation}
 which is by construction skew-symmetric and satisfies Leibnitz rule.
The Jacobi identity in local coordinates is:
$$
\begin{array}{c}
\pi^{ij} \, \partial_j \pi^{kl}  \,\partial_j f  \,\partial_k g  \,\partial_l h +
\pi^{ij}  \,\partial_j \pi^{kl}  \,\partial_j g  \,\partial_k h  \,\partial_l f +
\pi^{ij}  \,\partial_j \pi^{kl}  \,\partial_j h  \,\partial_k f  \,\partial_l g = 0 \\
\Updownarrow\\
\pi^{ij} \, \partial_j \pi^{kl} \; \partial_i \wedge \partial_k \wedge \partial_l =0 
\end{array}
$$
The last line is equivalent to
$$
[\pi,\pi]_S=0
$$ 
Recalling that $\mathfrak{g}_S^{\bullet}(M)$ is a DGLA with $d=0$, it is evident that
 Poisson bivector fields are  the solutions to the 
 equation \ref{MC-eq}  on $\mathfrak{g}_S^{\bullet}(M)$
\begin{equation} \label{MC-eq-calV}
d \pi + \frac{1}{2} [\pi,\pi]_S = 0, \qquad \pi \in \mathfrak{g}_S^{1}(M). 
\end{equation}

Finally, formal Poisson brackets $\lbrace \cdot,\cdot\rbrace_{\hbar}$ are associated to a formal 
bivector $\pi_{\hbar}\in\mathfrak{g}_S^{1}(M)\llbracket \hbar  \rrbracket$ as in \eqref{formal-poisson} 
and the gauge group action is defined in \ref{gauge-action}.

\subparagraph{Multidifferential operators}
Now we discuss the subalgebra of the Hochschild DGLA of multidifferential operators. Recall that
the Hochschild complex of an associative unital algebra $A$ is the complex $\tilde{C}(A,A)$ with vanishing components in degree $n<0$ 
and whose $n$-th component, for $n\geq 0$ is the space
\begin{equation}
\tilde{C}(A,A):= \sum_{n=-1}^{\infty}\tilde{C}^n(A,A)\qquad     \tilde{C}^n(A,A)=Hom(A^{\otimes n},A).
\end{equation}
By definition, the differential of a $n$-cochain $f$ is the $(n+1)$-cochain defined by
\begin{equation}\label{diff}
\begin{split}
&(-1)^n(df)(a_0,\dots, a_n) = a_0 f(a_1,\dots, a_n)-\\
&\sum_{i=0}^{n-1}(-1)^i f(a_0,\dots, a_i a_{i+1},\dots, a_n)
                                         +(-1)^{n-1}f(a_0,\dots, a_{n-1})a_n
\end{split}
\end{equation}
The Hochschild cohomology $H(A,A)$ of $A$ is the cohomology associated to the Hochschild complex $Ker\; d/Im\; d$.
The normalized 
Hochschild complex is 
\begin{equation}\label{eq: hcc}
C^n(A,A)=Hom(\bar{A}^{\otimes n},A)
\end{equation}
where $\bar{A}=A/k1$.
Now we introduce a new structure on the Hochschild complex, the Gerstenhaber bracket \cite{G}.
The Gerstenhaber product of $f\in\tilde{C}^n(A,A)$ and $g\in \tilde{C}^m(A,A)$ is the $(n+m-1)$-cochain defined by
\begin{equation}\nonumber
(f\circ g)(a_1,\dots, a_{n+m-1})=\sum_{j=0}^{n-1}(-1)^{(m-1)j}f(a_1,\dots, a_j, g(a_{j+1}, \dots, a_{j+m} ),\dots)
\end{equation}
that is not associative in general. As a consequence, we define the \textbf{Gerstenhaber bracket} as follows:
\begin{equation}\label{gb}
[D,E]_G=D\circ E-(-1)^{(n-1)(m-1)}E\circ D.
\end{equation}
We can easily check that it satisfies the (graded) Jacobi identity. We notice now  that the Hochschild differential can be expressed in terms of the Gerstenhaber bracket 
and the multiplication $m:A\otimes A\rightarrow A$ of $A$ as
\begin{equation}
d=[m, \cdot]_G:\tilde{C}^{\bullet}(A,A)\rightarrow \tilde{C}^{\bullet+1}(A,A)
\end{equation}
The space $C^{\bullet}(A,A)$ endowed with the Gerstenhaber bracket (\ref{gb}) and the differential (\ref{diff}) is a DGLA, called Hochschild DGLA. If $A=C^{\infty}(M)$, we are interested to a particular DGL subalgebra: the DGLA of multidifferential operators $\mathcal{D}$, so we consider only the maps from $A^{\otimes n}$ to $A$ which are multi-differential. More precisely,
\begin{equation}
\mathcal{D}:=\oplus_i \mathcal{D}^i
\end{equation}
 where $ \mathcal{D}^i$ are the subspaces of Hochschild $H^i(C^{\infty}(M),C^{\infty}(M))$ consisting of differential operators acting on $C^{\infty}(M)$. $\mathcal{D}$ is closed under $[\cdot,\cdot]_G$ and under the action of $d$, hence it is a DGL subalgebra of Hochschild DGLA.
 
 Because of Definition \ref{def-star} we are interested to a particular class of differential operators. Remember that the
  requirement of the star product $f\star g= f\cdot g+\sum_{n=1}^{\infty}\hbar^n P_n(f,g)$ to preserve the unit of the algebra
   implies $P_n(f,1)=P_n(1,f)=0$ for any $n>1$. This means we are interested to differential operators vanishing on constant
    functions. With this restriction we get a new DGLA, $\mathfrak{g}_G^{\bullet}(C^{\infty}(M))\subset \mathcal{D}$.

Finally we observe that the associativity of the product $m$ can be written in terms of $[\cdot,\cdot]_G$:
\begin{equation}
\begin{split}
[m,m]_G(f,g,h)&= \sum_{i=0}^1(-1)^i(m\circ_i m)(f,g,h)-\\
&-(-1)^1\sum_{i=0}^1(-1)^i(m\circ_i m)(f,g,h)=\\
&=2(m(m(f,g),h)-m(f,m(g,h)))=0
\end{split}
\end{equation}
Given an element $P \in \mathfrak{g}_G^{1}(C^{\infty}(M))$, we can interpret $m+P$ as a deformation
of the original product. As showed above, the associativity
of $m+P$ reads
$$
[m+P,m+P]_G=0.
$$
Observe that, since $m$ is associative and 
$[m,P]_G=[P,m]_G=dP$
the requirement of associativity of the deformed product $m+P$ can be rewritten exactly as a Maurer--Cartan equation \eqref{MC-eq}
\begin{equation} \label{MC-eq-calD}
dP+ \frac{1}{2} [P,P]_G =0.
\end{equation}
Since $P \in \hbar\mathfrak{g}_G^{\bullet}(C^{\infty}(M))\llbracket \hbar  \rrbracket$
is a formal sum of bidifferential operators,
we introduce
 the formal counterpart of the DGLA $\mathfrak{g}_G^{\bullet}(C^{\infty}(M))$ and 
the deformed product satisfies the requirements of star product (see Definition \ref{def-star}). The gauge group
is given by formal differential operators and the action on the star 
product is given by \ref{iso}.

\subsection{Kontsevich formality theorem}

A we mentioned above, Kontsevich's main result is that the DGLA $\mathfrak{g}_G^{\bullet}(C^{\infty}(M))$ is formal. This
 result relies on the existence of a previous result by Hochschild, Kostant and Rosenberg \cite{HKR} which establishes the
  existence of an isomorphism between the cohomologies of the algebra of multidifferential operators and the algebra of
   multivector fields.
\begin{theorem}[Hochschild-Kostant-Rosenberg \cite{HKR}]\label{hkr}
The formula
\begin{equation}
D_{\pi}(a_1,\dots, a_n)=\langle\pi, da_1\dots da_n\rangle
\end{equation}
defines a quasi-isomorphism
\begin{equation}
(\Gamma(T,\wedge^{\bullet}T),0)\rightarrow C^{\bullet}(C^{\infty}(M),C^{\infty}(M))
\end{equation}
In particular, the cohomology groups $ H^{\bullet}(C^{\infty}(M),C^{\infty}(M))$ is isomorphic to
\begin{equation}
\Gamma(T,\wedge^{\bullet}T),
\end{equation}
where the bracket induced by $[\cdot,\cdot]_G$ becomes the Schouten bracket $[\cdot,\cdot]_S$.
\end{theorem}
In order to introduce the formality theorem we need the notion of $L_{\infty}$-quasi isomorphism. Here we don't discuss the
 definition of $L_{\infty}$-algebras (see \cite{CI}).
\begin{definition}
Let $\mathfrak{g}_1$ and $\mathfrak{g}_2$ be two DGLA. By definition, an $\mathbf{L_{\infty}}$-\textbf{morphism} $f:\mathfrak{g}_1\rightarrow\mathfrak{g}_2$ is given by a sequence of maps
\begin{equation}
f_n:\mathfrak{g}_1^{\otimes n}\rightarrow \mathfrak{g}_2,\quad n\geq 1,
\end{equation} 
homogeneous of degree $1-n$ and such that the following conditions are satisfied:
\begin{enumerate}
	\item The morphism $f_n$ is graded antisymmetric, i.e. we have
\begin{equation}\nonumber
f_n(x_1,\dots, x_i,x_{i+1},\dots x_n)=-(-1)^{\vert x_i\vert\vert x_{i+1}\vert}f_n(x_1,\dots, x_{i+1},x_{i},\dots x_n)
\end{equation}
for all homogeneous $x_1,\dots, x_n$ of $\mathfrak{g}_1$. 
	\item We have $f_1\circ d=d\circ f_1$ i.e. the map $f_1$ is a morphism of complexes.
	\item $f_1$ is compatible with the brackets up to a homology given by $f_2$. In particular, $f_1$ induces a morphism of graded 
	Lie algebras from $H^{\bullet}(\mathfrak{g}_1)$ to $H^{\bullet}(\mathfrak{g}_2)$.
	\item More generally, for any homogeneous element $x_1,\dots, x_n$ of $\mathfrak{g}^{\bullet}$,
\begin{equation}
	\begin{split}
		\sum \pm & f_{q+1}([x_{i_1},\dots, x_{i_p}]_p,x_{j_1},\dots, x_{j_q})= \\
& \sum\pm \frac{1}{k!}[f_{n_1}(x_{i_{11}},\dots, x_{i_{1n_1}}),\dots,f_{n_k}(x_{i_{k1}},\dots, x_{i_{kn_k}})]
	\end{split}
\end{equation}
\end{enumerate} 
\end{definition}
Roughly, an $L_{\infty}$-morphism is a map between DGLA which is compatible with the brackets up to a given coherent system
of higher homotopies. 
\begin{definition}
An $\mathbf{L_{\infty}}$-\textbf{quasi isomorphism} is an $L_{\infty}$-morphism whose first components is a 
quasi-isomorphism.
\end{definition}
Kontsevich's Formality theorem can be finally stated as follows:
\begin{theorem}[Kontsevich \cite{Ko}]
There exists natural $L_{\infty}$-quasi isomorphism
\begin{equation}
K:\mathfrak{g}_S^{\bullet}(M)\rightarrow \mathfrak{g}_G^{\bullet}(C^{\infty}(M))
\end{equation}
The component $K_1$ of $K$ coincides with the quasi-isomorphism defined in the Hochschild-Kostant-Rosenberg Theorem \ref{hkr}.
\end{theorem}
Kontsevich's formality map  proves the one-to-one correspondence between the equivalence class of formal Poisson structures on $M$ and the isomorphism class of star products on $C^{\infty}(M)$. Remember from Section \ref{gla} that for any DGLA $(\mathfrak{g}, [\cdot,\cdot], d)$ we defined the set of solutions of the Maurer-Cartan equation
$$
MC(\mathfrak{g})=\lbrace a\in \hbar\mathfrak{g}^1\llbracket\hbar\rrbracket\;\ \vert \; da+\frac{1}{2}[a,a]=0\rbrace
$$
The group $exp(\hbar \mathfrak{g}^0\llbracket\hbar\rrbracket )$ acts on $MC(\mathfrak{g})$ by (\ref{gac}). Put
$$
M(\mathfrak{g})=MC(\mathfrak{g})/exp(\hbar \mathfrak{g}^0\llbracket\hbar\rrbracket )
$$
Now we introduce a new
\begin{theorem}[$L_{\infty}$-quasi isomorphism theorem]
Given a $L_{\infty}$-quasi isomorphism $f:\mathfrak{g}_1\rightarrow\mathfrak{g}_2$, this induces an isomorphism
$$
M(\mathfrak{g}_1)\simeq M(\mathfrak{g}_2)
$$
\end{theorem}
Consider the DGLA's of multidifferential operators and multivector fields. From Section \ref{gla} follows immediately that  $M(\mathfrak{g}_S^{\bullet}(M))$ is the set of equivalence classes of formal Poisson structures $\pi_{\hbar}$  and $M(\mathfrak{g}_G^{\bullet}(C^{\infty}(M)))$ is the set of isomorphism classes of star products on $C^{\infty}(M)$. Using the Kontsevich theorem and the $L_{\infty}$-quasi isomorphism theorem follows that $M(\mathfrak{g}_S^{\bullet}(M))\simeq M(\mathfrak{g}_G^{\bullet}(C^{\infty}(M)))$. This proves Theorem \ref{thm:k}.

\bibliographystyle{plain}
\bibliography{/Users/macbook/Documents/Work/Research/Bibliography/references1}
 
\end{document}